\documentclass{article} 
\usepackage{amsthm,amsmath,amssymb,amsfonts,authblk}
\usepackage{fullpage}
\usepackage[bookmarks=false, pdfstartview=FitH]{hyperref}

\usepackage{epsfig}
\usepackage{url}

\newtheorem{theorem}{Theorem}

\newtheorem{lemma}{Lemma}[section]

\newtheorem{corollary}{Corollary}[section]

\newtheorem{definition}{Definition}

\title{On Resource-bounded versions of the van Lambalgen theorem}
\author[1]{Diptarka Chakraborty\thanks{diptarka@iuuk.mff.cuni.cz}\footnote{The author was supported
    by the European Research Council under the European Union's Seventh Framework Programme (FP/2007-2013)/ERC Grant Agreement n. 616787.}}
\author[2]{Satyadev Nandakumar\thanks{satyadev@cse.iitk.ac.in}}
\author[3]{Himanshu Shukla\thanks{hshukla.math04@gmail.com}}
\affil[1]{Computer Science Institute of Charles University,
  Malostransk{\'e}  n{\'a}mesti 25,
  118 00 Praha 1, Czech Republic } 
  \affil[2,3]{Department of Computer Science and Engineering,
  Indian Institute of Technology Kanpur, 
  Kanpur, Uttar Pradesh, India.}

\newcommand{\N}{\mathbb{N}} 
\newcommand{\R}{\mathbb{R}}
\newcommand{\Q}{\mathbb{Q}}
\newcommand{\Rplus}{[0, \infty)}
\newcommand{\PP}{\mathcal{P}}
\newcommand{\QQ}{\mathcal{Q}}
\newcommand{\res}{\upharpoonright}

\newcommand{\mmid}{\quad\mid\quad}

\usepackage[svgnames]{xcolor}

\begin{document}
\maketitle

\begin{abstract}
The van Lambalgen theorem is a surprising result in algorithmic
information theory concerning the symmetry of relative randomness. It
establishes that for any pair of infinite sequences $A$ and $B$, $B$
is Martin-L\"of random and $A$ is Martin-L\"of random relative to $B$
if and only if the interleaved sequence $A \uplus B$ is Martin-L\"of
random. This implies that $A$ is relative random to $B$ if and only if
$B$ is random relative to $A$ \cite{vanLambalgen}, \cite{Nies09},
\cite{HirschfeldtBook}. This paper studies the validity of this
phenomenon for different notions of time-bounded relative
randomness. 

We prove the classical van Lambalgen theorem using martingales and
Kolmogorov compressibility. We establish the failure of relative
randomness in these settings, for both time-bounded martingales and
time-bounded Kolmogorov complexity. We adapt our classical proofs when
applicable to the time-bounded setting, and construct counterexamples
when they fail. The mode of failure of the theorem may depend on the
notion of time-bounded randomness.
\end{abstract}

\section{Introduction}
In this paper, we explore the resource-bounded versions of van
Lambalgen's theorem in algorithmic information theory. van Lambalgen's
theorem deals with the symmetry of relative randomness. The theorem
states that an infinite binary sequence $B$ is Martin-L\"of random and
a sequence $A$ is Martin-L\"of random relative to $B$ if and only if
the interleaved sequence $A_0 B_0 A_1 B_1 \dots$ is Martin-L\"of
random \cite{vanLambalgen}. It follows that $A$ is Martin-L\"of random
relative to $B$ if and only if $B$ is Martin-L\"of random relative to
$A$.

This result is quite surprising, since it connects the randomness of
$A$ with the computational power $A$ possesses \cite{Nies09},
\cite{HirschfeldtBook}. This contrasts with relative computability -
for instance, every computably enumerable language is computable given
the halting problem as an oracle, but the halting problem is not
computable given an arbitrary c.e. language. Symmetry of relative
randomness is desirable for any robust notion of randomness. However,
we now know that it fails in several other settings - both Schnorr
randomness and computable randomness exhibit a lack of symmetry of
relative randomness \cite{Yu07, bauwens2015uniform}.

We explore whether this symmetry holds when Martin-L\"of randomness is
replaced with time-bounded randomness. Considering the failure of the
analogies of the van Lambalgen's theorem in many settings, it is
natural to guess that such a resource-bounded version of van
Lambalgen's theorem is false. Indeed, the existence of \emph{one-way
  functions} \cite{Goldreich01} from strings to strings which are easy
to compute but hard to invert, can be expected to have some bearing to
the validity of the resource-bounded van Lambalgen's theorem. In the
context of polynomial-time compressibility, Longpr\'{e} and Watanabe
\cite{LongpreWatanabe95} establish the connection between
polynomial-time symmetry of information and the existence of one-way
functions, and analogously, Lee and Romaschenko
\cite{LeeRomaschenko05} establish the connection for $CD$ complexity
\cite{LiVi08}.

Modern proofs of van Lambalgen's theorem proceed by defining Solovay
tests (see \cite{HirschfeldtBook}, \cite{Nies09}). The notion of a
resource-bounded Solovay test has not been studied, while the notion
of resource-bounded martingales \cite{Lutz:RM} and resource-bounded
Kolmogorov complexity have been studied extensively (see Allender
et. al. \cite{ABKMR06}). We approach the classical van Lambalgen's
theorem using prefix-free incompressibility and martingales, inspired
by the Solovay tests. This part may be of independent interest. We
then attempt to adapt these proofs to resource-bounded settings.

Our main results are the following. Let $t$ be a superlinear time
bound, and $t^X$ denote $t$-computable functions with oracle access
to the sequence $X$. Let $A \uplus B$ denotes the interleaving of $A$
and $B$.

\begin{enumerate}
\item \label{item1} Using the notion of $t$-bounded martingales, we
  show that there are $t$-nonrandom $A \uplus B$ where $B$ is
  $t$-random and $A$ is $t^B$-random. This result is unconditional,
  and analogous to the result of Yu \cite{Yu07}.
\item \label{item2} 
\begin{enumerate}
\item There are $t$-random sequences $A$ and $B$ where $A$ is
  $t^B$-nonrandom, but $B$ is $t^A$-random. However for this pair, $A
  \uplus B$ is still $t$-nonrandom. Thus the randomness of the
  interleaved sequence and mutual relative randomness of the pair are
  distinct notions for time-bounded martingales. 
  \item We establish a sufficient condition under which a $t$-random 
  $B$ and a $t^B$-nonrandom $A$ could still create $t$-random $A 
  \uplus  B$. This involves a non-invertibility condition reminiscent
  of one-way functions.
\end{enumerate}
\item \label{item3}There are $t$-compressible $A \uplus B$ such that
  $B$ is $t$-incompressible and $A$ is $t$-incompressible relative to
  $B$. This is an unconditional result analogous to~\ref{item1}.
\item \label{item4}If $B$ is $t$-compressible or $A$ is
  $t$-compressible with respect to $B$, then $A \uplus B$ is
  $t$-compressible. This is in contrast to~\ref{item2}.
\end{enumerate}
Thus van Lambalgen's theorem fails in resource-bounded
settings. Surprisingly, the manner of failure may depend on the
formalism we choose.  

The results in the paper also provide indirect evidence that
resource-bounded randomness may vary depending on the formalism. In
particular, the set of sequences over which resource-bounded
martingales fail may not be the same as the set of resource-bounded
incompressible sequences. The results in~\ref{item2} and~\ref{item4}
provide us a conditional separation between these two formalisms in 
case of resource-bounded settings.

The manner of failure in~\ref{item2} has to do with the oracle access
mechanism, and the proof hinges on a technical obstacle which may be
tangential to time-bounded computation. In the final section of the
paper, we propose a modified definition which we call $t$-bounded
``lookahead'' martingales with which we are able to show that if $B$
is $t$-lookahead-nonrandom or $A$ is $t$-lookahead-nonrandom relative
to $B$, then $A \uplus B$ is $t$-lookahead nonrandom. Here, the van
Lambalgen property for $t$-lookahead martingales fails in precisely
the same manner as $t$-incompressibility. This may be a reasonable
model to study resource-bounded martingales.

\section{Preliminaries}
We assume familiarity with the basic notions of algorithmic randomness
at the level of the initial chapters in Downey and Hirschfeldt
\cite{HirschfeldtBook} or Nies \cite{Nies09}.

We use the notation $\N$ for the set of natural numbers, $\Q$ for
rationals, and $\R$ for reals. We work with the binary alphabet
$\Sigma = \{0,1\}$. We denote the set of finite binary strings as
$\Sigma^{*}$ and the set of infinite binary sequences as
$\Sigma^{\infty}$. Finite binary strings will be denoted by lower-case
Greek letters like $\sigma$, $\rho$ etc. and infinite sequences by
upper-case Latin symbols like $X$, $Y$ etc. The length of a string
$\sigma$ is denoted by $|\sigma|$. The letter $\lambda$ stands for the
empty string. For finite strings $\sigma$ and $\rho$ and any infinite
sequence $X$, $\sigma \preceq \rho$ and $\sigma \preceq X$ denote that
$\sigma$ is a prefix of $\rho$ and $X$ respectively.

The substring of length $n$ starting from the $m^\text{th}$ position
of a finite string $\sigma$ or an infinite sequence $X$ is denoted by
$\sigma[m \dots m+n-1]$ and $X[m \dots m+n-1]$, where $m+n-1 <
|\sigma|$. When $m$ is 0, \emph{i.e.}  the first position, we
abbreviate the notation as $\sigma \res n$ and $X \res n$ -
e.g. $\sigma \res n$ is $\sigma[0 \dots n-1]$.

The concatenation of $\sigma$ and $\tau$ is written as
$\sigma \tau$.The notation $A \uplus B$ stands for the sequence
we get by interleaving the bits in $A$ with the bits in $B$,
\emph{i.e.} $A_0 B_0 A_1 B_1 \dots$.\footnote{It is also common to use
  $\oplus$, but we want to avoid confusion with the bitwise xor
  operation.}

A set of finite strings $S$ is said to be \emph{prefix-free} if no
string in $S$ can be a proper prefix of another string in $S$.

\begin{theorem}{(van Lambalgen, 1987) \cite{vanLambalgen}}
\label{thm:cevlt}
For any two infinite sequences $A$ and $B$, $B$ is Martin-L\"of random
and $A$ is Martin-L\"of random relative to $B$ if and only if
$A \uplus B$ is Martin-L\"of random.
\end{theorem}

\section{A Proof using Incompressibility}
We now prove Theorem \ref{thm:cevlt} via incompressibility
notions. Throughout the remainder of the paper, we fix a canonical set
of prefix-free codes for partial computable functions by
$\mathcal{P}$.
\begin{definition}
The \emph{self-delimiting Kolmogorov complexity} of $\sigma \in
\Sigma^*$ is defined by $ K(\sigma) = \min \{ |\pi| \mid \pi \in \PP
\text{ outputs } \sigma\}$.

Similarly, the \emph{conditional Kolmogorov complexity} of $\sigma \in 
\Sigma^*$ given $\tau \in \Sigma^*$ is defined by $ K(\sigma 
\mid \tau) = \min \{ |\pi| \mid \pi \in \PP \text{ outputs } \sigma 
\text{ on input } \tau \}$.
\end{definition}

Using the notion of incompressibility, it is well-known that we can
formulate an equivalent definition of random sequences \cite{Nies09}.
\begin{definition}
An infinite binary sequence $A$ is said to be \emph{incompressible} if
$\exists c \quad \forall n \quad K(A \res n) \ge n - c$.  The sequence
$A$ is \emph{incompressible with respect to} another binary sequence
$B$ (or $B$-incompressible) if $\exists c \quad \forall n \quad K^B(A
\res n) \ge n - c$. 
\end{definition}
The set of Martin-L\"of random sequences are precisely the set of
incompressible sequences. Relativizing the same result, the set of
Martin-L\"of random sequences relative to a sequence $B$ is
precisely the set of sequences incompressible with respect to $B$.

We now prove van Lambalgen's theorem using incompressibility. When we
consider the issue of resource-bounded van Lambalgen's theorems, we
try to either adapt these proofs where applicable, or examine the
issues which prevent such an adaptation. We prove the two directions 
of the van Lambalgen's theorem separately so as to emphasize the 
issues which arise in the resource-bounded setting.

The proof of the first direction relies on a form of Symmetry of
Information, a result first established by Levin and G\'{a}cs
\cite{LiVi08}. To this end, we mention basic results from the theory
of self-delimiting (prefix-free) Kolmogorov complexity.

\begin{definition}
A computably enumerable set $L \subseteq \Sigma^* \times \N$ is said
to be a \emph{bounded request set} if $\sum_{(\sigma,n) \in L}
\frac{1}{2^{n}} \le 1$.
\end{definition}

We may view each element $(w,n)$ as a request to encode $w$ using at
most $n$ bits. The boundedness condition is a promise that the
requested code lengths satisfy the Kraft inequality. The Machine
Existence Theorem states that there is some prefix-free code which can
satisfy all requests in a bounded request set.

\begin{theorem}{(Machine Existence Theorem)\cite{Nies09}}
Let $L$ be a bounded request set. Then there is a prefix-free set of
codes $\PP$ which, for each $(y,m) \in L$, allocates a prefix-free
code $\tau \in \Sigma^m \cap \PP$ for $y$.
\end{theorem}

The coding theorem relates the algorithmic probability of a string to
its prefix-free Kolmogorov Complexity. We state it here in the form
applicable to pairs of strings, but an analogous result holds for
strings.

\begin{theorem}{(Coding Theorem)\cite{Nies09}}
Let $\tau$ be a finite string. Let $\PP$ be a prefix-free encoding of
partial-computable functions outputting pairs of strings. Denote
$\PP_\tau \subseteq \PP$ as the set of prefix-free codes which output
pairs $(\sigma,\tau)$ for some arbitrary string $\sigma$. Then there
is a constant $c$ such that 
\begin{align*}
2^{c-K(\sigma,\tau)} > \sum_{\rho \in \PP \text{ outputting }
  (\sigma,\tau)} 2^{-|\rho|}
\end{align*}
\end{theorem}

Using these, we now state and prove the variant of ``Symmetry of
Information'' which we use to establish Lemma
\ref{lem:incompressibility_soi}. 

\begin{lemma}
\label{lem:symmetry_of_information}
Let $\sigma$ be a finite string with $K(\sigma)>|\sigma|-c$, and
$\tau$ be a finite string. Then $|\sigma| + K(\tau|\sigma) \le
K(\sigma,\tau) + O(1)$.
\end{lemma}
\begin{proof}
Let $p_i$ be an arbitrary program in the computable enumeration of
$\PP$, the set of programs which output string pairs. Consider the
program $R_{p_i}$ which can be generated from $p_i$, defined by the
following algorithm.\\
1. Input $\sigma$.\\
2. Let $U(p_i)$ output the string pair $(\alpha,\tau)$.\\
3. If $\alpha$ is equal to $\sigma$, then we output $(\tau, |p_i| -
  |\sigma| + c')$, where $c'$ satisfies the inequality below.\\
Corresponding to the computable enumeration $p_1$, $p_2$, $\dots$ of
$\PP$, we obtain a computable enumeration $R_{p_1}$, $R_{p_2}$,
$\dots$. We now show that this forms a valid enumeration of a
\emph{bounded request set} (see, for example, \cite{Nies09} page 78).

Let $N_{\sigma}$ be the set of indices $i \in \PP$ where $U(p_i)$
outputs a pair of strings of the form $(\sigma,\tau)$ for some
$\tau$. First, we have
\begin{align*}
\sum_{i \in N_{\sigma}} \frac{1}{2^{|p_i|-|\sigma|+c'}} 
< 2^{|\sigma|-c'} \sum_{i \in N_{\sigma}} \frac{1}{2^{|p_i|}} 
< 2^{|\sigma|-c'+c''} \sum_{\tau \in \Sigma^*}
   \frac{1}{2^{K(\sigma,\tau)}} 
= \frac{2^{|\sigma|-c'+c''}}{2^{K(\sigma)}}
< 1,
\end{align*}
where the second inequality follows from the Coding Theorem (see, for
example, Nies \cite{Nies09}, Theorem 2.2.25), and the last inequality
follows from the assumption.

Hence $R_{p_1}$, $R_{p_2}$, $\dots$ is a computable enumeration of a
bounded request set. By the Machine existence theorem for prefix-free
encoding (see for example, \cite{Nies09} Theorem 2.2.17), it follows
that for any request $(\tau,|p_i|-|\sigma|+c')$, there is a prefix-free
encoding of $\tau$ given $\sigma$ which has length
$|p_i|-|\sigma|+c'$. Now, consider a shortest prefix-free code $p_i$
for $(\sigma,\tau)$. We have that $|p_i| = K(\sigma,\tau)$. Hence
$K(\tau~\mid~\sigma) \le K(\sigma,\tau) - |\sigma|+O(1)$.
  \end{proof}
\begin{lemma}
\label{lem:incompressibility_soi}
If $B$ is incompressible and $A$ is $B$-incompressible, then 
$A \uplus B$ is incompressible.
\end{lemma}
\begin{proof}
Suppose that for every $n$, $K(B \res n) \ge n-c$ and $K^B(A \res n)
\ge n - c'$. This implies that $K(A \res n \mid B \res (n-1)) \ge
n-c'$. By the version of the Symmetry of information in Lemma
\ref{lem:symmetry_of_information}, we have
$$ (2n-1)-c'\quad<\quad (n-1)+K^B(A \res n)\quad\le\quad K((A \uplus
B)\res (2n-1))+O(1).$$   
A similar argument works for $K((A \uplus B)\res 2n)$. This completes
the proof. 
  \end{proof}
The above proof relied on symmetry of information of prefix-free
Kolmogorov Complexity. Since reasonable complexity-theoretic
hypotheses imply that this fails in resource-bounded settings, we can
foresee that this direction fails in resource-bounded settings, as we
show in section \ref{secn:rbi}.

Since the first direction was a consequence of Symmetry of
Information, it is reasonable to expect the converse direction to
follow from the subadditivity of $K$: $ K((A \uplus B)\res
2n)\quad\le\quad K(B \res n) + K^B(A \res n)+O(1)$. However, this
runs into the following obstacle.  If the prefix of $B$ is
compressible with complexity, say $n-\log(K(n))$, and the prefix of
$A$ is $B$-incompressible with conditional complexity $n+ K(n)$, then
we cannot conclude from subadditivity that
$K((A~\uplus~B)~\res~2n)$ is less than $2n$.  Thus concatenating the
shortest prefix-codes for $B \res n$ and $A \res n$ given $B \res n$
to obtain a prefix-free code for $(A \uplus B) \res 2n$ may be
insufficient for our purpose. We now show the converse direction
through more succinct prefix-free codes.
\begin{lemma}
\label{lem:combine_K_unbounded}
If $B$ is compressible or $A$ is $B$-compressible, then $A \uplus B$
is compressible.
\end{lemma}
\begin{proof}
Let $K(B \res n) < n-c$, and let $\sigma$ be a shortest program from
the c.e. set of codes $\PP$ which outputs $B \res n$. Consider the
prefix-free set defined by
\begin{align}
\mathcal{Q}_n = \{\tau \rho\ \mid\ \tau \in \PP, |\rho|=n \}.
\end{align}
This is a prefix-free c.e. set of codes. Then $\sigma 
(A~\res~n)$ - \emph{i.e.} $\sigma$ concatenated with the first $n$ 
bits of $A$ - is a code for $A \uplus B$ for some machine $M$ which
first runs $R(\sigma)$ to output $B \res n$, then interleaves $A \res
n$ with $B \res n$ to produce $(A \uplus B) \res n$. The length of
this code is at most $K(B \res n) + n + O(1)$, showing that $A \uplus
B$ is compressible at length $2n$.

Now, assume that $A$ is $B$-compressible, and let $n$ and $m$ satisfy
$$K(A \res n \mid B \res m) \le n - c.$$ 
Since we can make redundant queries, without loss of generality, we
assume that $m\ge n$. Let $\PP$ be the set of prefix-free encodings of
one-argument partial computable functions. We construct a prefix-free
code to show that $(A \uplus B)$ is compressible at length
$2m$. Consider $\QQ_{m,n}$ defined by
\begin{align}
\label{prefix_set_mn}
\QQ_{m,n} = \{ \tau  \sigma \mmid \tau \in \PP,\ |\sigma| =
2m-n\}.
\end{align}
Since $\PP$ is a prefix-free set and we append strings of a fixed
length to the prefix-free codes, $\QQ_{m,n}$ is also a prefix-free
set. If $\PP$ is computably enumerable, then so is $\QQ_{m,n}$. 
Moreover, there is an encoding of $A_0 B_0 \dots A_{m-1} B_{m-1}$ in 
$\QQ_{m,n}$ given by $\alpha  (B\res m)  (A[n \dots m-1])$. This
encoding has length at most $n - c + m + m - n$, which is at most 
$2m-c$. 
  \end{proof}

We may expect this proof to be easily adapted to resource-bounded
settings. Inherent in the above proof is the concept of universality
-- since there is a universal self-delimiting Turing machine which
incurs at most additive loss over any other prefix-free encoding, it
suffices to show that there is some prefix-free succinct encoding. We
appropriately modify this in resource-bounded settings which lack such
universal machines in general. 

\section{Martingales and van Lambalgen's Theorem}
We now approach van Lambalgen's theorem using martingales,
adapting the Solovay tests in the literature \cite{Nies09},
\cite{HirschfeldtBook}.
\begin{definition}
A function $d: \Sigma^* \to \Rplus$ is said to be a \emph{martingale}
if $d(\lambda) = 1$ and for every string $w$, $d(w) =
(d(w0)+d(w1))/2$, and a \emph{supermartingale} if for every string
$w$, $d(w) \ge (d(w0)+d(w1))/2$.

A \emph{martingale} or a \emph{supermartingale} is said to be
\emph{computably enumerable} (c.e.) if there is a Turing Machine $M :
\Sigma^* \times \N \to \Q$ such that for every string $w$, the
sequence $M(w,n)$ monotonically converges to $d(w)$ from below.
\end{definition}
The rate of convergence in the above definition need not be
computable.
\begin{definition}
We say that a martingale $d$ \emph{succeeds} on $X \in
\Sigma^{\infty}$ if $\limsup_{n \to \infty} \\d(X~\res~n) = \infty$,
written $X \in S^\infty[d]$, and that $d$ \emph{strongly succeeds} on
$X$, written $X \in S^\infty_\text{str}[d]$, if $\liminf_n d(X~\res~n)
= \infty$.

If no computably enumerable martingale or supermartingale succeeds on
$X$, then we say that $X$ is Martin-L\"of random. We say that $X$ is
\emph{non-Martin-L\"of random} relative to $Y$ if there is a
computably enumerable oracle martingale $d$ such that $\limsup_{n \to
  \infty} d^Y(X \res n) = \infty$.
\end{definition}
\begin{lemma}
If $B$ is not Martin-L\"of random or $A$ is not Martin-L\"of random
relative to $B$, then $A \uplus B$ is not Martin-L\"of random.
\end{lemma}
\begin{proof}
Let $d_B$ be a martingale that succeeds on $B$. Then the martingale
$d_{AB}$ defined by setting $d_{AB}(\lambda)$ to 1 and
\begin{align}
\label{eqn:martingale_B_AB}
d_{AB}(\sigma_0 \tau_0 \dots \tau_{n-2} \sigma_{n-1}) &= 
 d_{AB}(\sigma_0 \tau_0 \dots \tau_{n-2}).\\
\notag
d_{AB}(\sigma_0 \tau_0 \dots \sigma_{n-1} \tau_{n-1}) &= 
d_{B}(\tau_0 \dots \tau_{n-1}).
\end{align}
The above definition is a martingale since for any $n \ge 2$,
$$d_{AB}(\alpha_0 \beta_0 \dots \beta_{n-2} \alpha_{n-1})=
d_{B}(\beta_0 \dots \beta_{n-2}).$$ 
Clearly, $\limsup_{n \to \infty} d_{AB}(A \uplus B) =
\limsup_{n\to\infty} d_{B}(B)$ and hence $d_{AB}$ succeeds on $A
\uplus B$. 

Now, suppose $d$ succeeds on $A$ given oracle access to $B$. Consider
martingale $m$ defined by setting $m(\lambda)$ to $1$ and setting
\begin{align*}
  \ m(\sigma_0 \tau_0 \dots \sigma_{n-1})[s]
  &= d^{\tau \res s}(\sigma \res i)\\
  m(\sigma_0 \tau_0 \dots \tau_{n-1})[s] &= 
  m(\sigma_0 \tau_o \dots \sigma_{n-1})[s], 
\end{align*}
where the notation $m(\alpha)[s]$ denotes the value that the
computation $m$ assigns to $\alpha$ at stage $s$ and for any string
$x \in \Sigma^*$, the value of $m(x)=\limsup_{s \to
  \infty}m(x)[s]$. Note that in the computation of $d$ in the second
step, each fixed initial segment of $v$ can query longer initial
segments of $w$ when they become available.

Since $d$ is a c.e. oracle martingale, it follows that $m$ is a
c.e. martingale. For every pair of infinite sequences $V$ and $W$ and
for every $l$, there is a number $n$ computable from $V \res l$ and
$W$ such that for all large enough stages $s$,
$d^{W \res s}(V \res l) = d^W(V \res l)$. Thus for each $l$, the value
of $m((V \res l) \uplus (W \res l))[s]$ is the same as
$m((V \res l) \uplus (W \res l))[s_1]$ for all $s_1 > s$, for some
large enough $s$. It follows that $m$ is c.e. martingale. Since $d^B$
succeeds on $A$, $m$ succeeds on $A\uplus B$ and this completes the
proof.   \end{proof}

The converse also holds. However, in the latter part of this paper we
show that the analogous results may not hold in time-bounded versions.
\begin{lemma}
\label{lem:martingaleconverse}
If $A \uplus B$ is not Martin-L\"of random, then either $B$ is not
Martin-L\"of random or $A$ is not Martin-L\"of random relative to
$B$. 
\end{lemma}
For the proof of this lemma, it is convenient to use a notion which is
related to martingales.

\begin{definition}
A function $f: \Sigma^\infty \to [0,\infty]$ is called \emph{lower
  semicomputable} if the set 
$$\{(\sigma, q) \mid 
\sigma \prec X, X \in \Sigma^\infty \text{ and }
q \le f(X)\}$$
is computably enumerable -- \emph{i.e.} the rational points in the
lower graph of $f$ is computably enumerable.
\end{definition}

\begin{definition}
A lower semicomputable function $f$ is said to be a \emph{measure of
  impossibility}\footnote{also called an \emph{integral test}} with
respect to a probability measure $P$ if $\int f dP < \infty$.
\end{definition}

We focus our attention on the uniform probability measure on $[0,1]$.
We have the following theorem characterizing Martin-L\"of randomness
in terms of measures of impossibility. 

\begin{theorem}
\label{thm:measureimposs}
A sequence $X \in \Sigma^\infty$ is Martin-L\"of random if and only if
for every measure of impossibility $f: \Sigma^\infty \to [0, \infty]$,
$f(X) < \infty$.
\end{theorem}

Relativizing the proof of the above theorem, we have the following.

\begin{corollary}
\label{thm:measureimposs_rel}
A sequence $X \in \Sigma^\infty$ is Martin-L\"of random relative to $Y
\in \Sigma^\infty$ if and only if for every measure of impossibility
$f^Y: \Sigma^\infty \to [0, \infty]$, $f^Y(X) < \infty$.
\end{corollary}

\begin{proof}[Proof of Lemma~\ref{lem:martingaleconverse}]
Suppose $d_{AB}$ is a martingale which succeeds on $A \uplus
B$. Define the martingale $d_B$ by setting $d_B(\lambda)=1$ and
\begin{align}
d_B(\sigma)=2^{-|\sigma|} \sum \limits_{\{\tau ~:~
  |\tau|=|\sigma|\}} d_{AB}(\tau \uplus \sigma)
\end{align}
for finite nonempty strings $\sigma$. It is easy to establish that
$d_B$ is c.e. if $d_{AB}$ is.

Suppose that for any positive $N$, there are infinitely many $n$ such
that
\begin{align}
\label{ineq:semi_martingale}
  \sum_{\sigma \in \Sigma^n} 
    d_{AB}(\sigma_0B_0 \dots \sigma_{n-1}B_{n-1})
  \ge N 2^n.
\end{align}
In this case, $d_B(B \res n)$ is at least $N$. Hence $d_B$ succeeds on
$B$, and the lemma holds.

Otherwise, there is some positive $N$ and an $n_0$ such that for all
$n \ge n_0$, we have the inequality in (\ref{ineq:semi_martingale})
reversed.

By using the ``savings account'' trick, we can define another
martingale which succeeds strongly on $A \uplus B$. We consider the
paired functions $(f,s): \Sigma^* \to \Rplus$ defined as
follows. Initially, $f(\lambda)=1$ and $s(\lambda)=1$. On any $\sigma
\in \Sigma^*$, $f(\sigma b)$ bets the same ratio of its capital as
$d_{AB}$, and if the resulting capital is greater than 2, then we set
$f(\sigma b)$ to 1 and transfer the remaining capital to $s(\sigma b)$.
\begin{align*}
f(\lambda) &= 1\\
f(\sigma b) &= \begin{cases}
  \frac{d_{AB}(\sigma b)}{d_{AB}(\sigma)}f(\sigma) &\text{if } 
  \frac{d_{AB}(\sigma b)}{d_{AB}(\sigma)}f(\sigma)<2\\  
  1                      &\text{otherwise,}
\end{cases}
\end{align*}
and
\begin{align*}
s(\lambda) &= 0\\
s(\sigma b) &= \begin{cases}
  s(\sigma) &\text{if }
  \frac{d_{AB}(\sigma b)}{d_{AB}(\sigma)}f(\sigma)<2\\  
  s(\sigma)
  +\left(\frac{d_{AB}(\sigma b)}{d_{AB}(\sigma)}f(\sigma)-1\right) 
  &\text{otherwise.}  
\end{cases}
\end{align*}
We can verify that $s$ is monotone increasing in its prefix length and
if $X \in S^\infty[d_{AB}]$, then $X \in S^\infty_\text{str}[s]$.
  
Now consider the function $g^Y:\Sigma^\infty \to [0,\infty]$ defined
by
$$g^Y(X) = \lim_{n \to \infty} s(X_0 Y_0 \dots X_{n-1}Y_{n-1}).$$ 
For every sufficiently large $n$, we have $\int_{C_\sigma}g^Y d\mu \le
N 2^{n}$, by assumption, where $C_\sigma \subseteq \Sigma^\infty$
consists of all infinite sequences with $\sigma$ as a prefix, and
$\mu$ is the Lebesgue measure on Cantor Space. As in \cite{Nan09}, we
can verify that is a measure of impossibility, using Fatou's lemma.

Hence $g^Y$ is a measure of impossibility. Now by
Theorem~\ref{thm:measureimposs_rel}, if $d_{AB}$ succeeds on $A \uplus
B$, then $g^B(A)=\infty$. Also, if $d_{AB}$ is computably enumerable,
then $g^B$ is lower semicomputable. Hence $A$ is not Martin-L\"of
random relative to $B$.   \end{proof}

\section{Resource-bounded relative randomness and incompressibility} 
\label{secn:rbi}

We consider time-bounded self-delimiting Kolmogorov complexity in this
section. While there are several variants of this notion (see
e.g. \cite{LongpreWatanabe95}, \cite{ABKMR06}), we deal with the
simplest one here. 

The time-bound is a function of the lengths of its
output\footnote{Considering time bound that is dependent on output
  length is not unnatural for decompressors. To make it input-length
  dependent it is customary to append $1^l$ as an additional input
  where $l$ is the output length} as in \cite{LongpreWatanabe95}. Through out this paper we will restrict ourselves only to the time-constructible time-bound function $t$. We
first fix a prefix-free set $\PP$ encoding the set of
partial-computable functions. We do not insist that $\PP$ consist
solely of functions which run in $t$ steps, since the results are
identical with or without this assumption.
\begin{definition}
The \emph{$t$-time-bounded complexity} of $\sigma$ is defined as
defined by 
\begin{align}
K_T(\sigma;t) = \min 
\{ |\pi| \mid  \pi \in \PP \text{ outputs } \sigma
                  \text{ in } \le O(t(|\sigma|)) \text{ steps}\},
\end{align}
and the \emph{conditional $t$-time-bounded complexity} of $\sigma$
given $\tau$ be defined by  
\begin{align}
K_T (\sigma \mid \tau;t) = \min 
\{ |\pi| \mid \pi \in \PP, \pi(\tau) \text{ outputs } \sigma 
                \text{ in } \le O(t(|\sigma|)) \text{ steps}\}.
\end{align}
\end{definition}
For any fixed time bound $t$, we do not have universal machines within
the class of $t$-bounded machines. However, there are invariance
theorems (see e.g. \cite{LiVi08} Chapter 7). Hence we can use the
definition of time-bounded complexity to define the notion of
incompressible infinite sequences.
\begin{definition}
An infinite binary sequence $X$ is said to be
\emph{$t$-incompressible} if $\exists c \quad \forall n \quad K_T(X
\res n;t) \ge n - c$ and \emph{$t^Y$-incompressible} if $\exists c
\quad \forall n \quad K_T(X \res n \mid Y \res m;t) \ge n - c$ for
some $m$ depending on the value of $X$ and $n$.
\end{definition}
If for $t'>t$, a sequence $X$ is $t'$-incompressible, then it is
$t$-incompressible as well. Moreover, for every $X$ and $n$, $K_T(X
\res n) \ge K(X \res n)$. Since the set of $K$-incompressible
sequences has measure $1$, we know that the set of $t$-incompressible
sequences has measure $1$ as well. When the time bound is understood
from the context, we write $K_T(\sigma)$ and $K_T(\sigma \mid \tau)$. 

We now show that for time-bounded Kolmogorov complexity, only one
direction of van Lambalgen's theorem holds. We first show that if we
can compress $B$, or $A$ relative to $B$ within the time bound, then
it is possible to compress $A \uplus B$ within the time bounds,
adapting the proof of Lemma \ref{lem:combine_K_unbounded}.

\begin{lemma}
\label{lem:t_incompressibility_one}
If $B$ is $t$-compressible or $A$ is $t^B$-compressible, then $A
\uplus B$ is $t$-compressible.
\end{lemma}
\begin{proof}
  Assume that $B$ is $t$-compressible. Then there is a constant $c$
  and infinitely many $n$ such that there is a short program in
  $\beta \in \PP$ with $|\beta| < n-c$ which outputs $B \res n$ within
  $O(t(n))$ steps.

  For any such $n$, consider the prefix-code defined by 
  \begin{align}
    \QQ_n = \{ \sigma \alpha \mid \sigma \in \PP, |\alpha| = n \}.
  \end{align}

  This forms a prefix encoding, containing a code $[\beta(A \res n)]$
  for $(A \uplus B) \res 2n$. Moreover, it is possible to decode
  $(A \uplus B) \res 2n$ from its code within $O(t(2n))$ steps.

  Suppose $A$ is $t^B$-compressible. Assume that $K_T((A \res n)
  \mid (B \res m);t) \le n-c$, witnessed by a code $\alpha$. Without
  loss of generality, we may assume $n \le m$. Then consider
  $\QQ_{m,n}$ as defined in (\ref{prefix_set_mn}). We see that
  $\QQ_{m,n}$ is a computably enumerable prefix set.  The code
  $(\alpha (B \res m)  A[n \dots m-1]) \in \QQ_{m,n}$ of $(A
  \uplus B) \res 2n$ can be decoded in time $O(t(2m))$, and is shorter
  than $2m-c$.

  Hence $K_T((A \uplus B) \res 2m) < 2m-c$. 
  \end{proof}

The converse of the above lemma is false. We do not appeal to the
failure of polynomial-time (in general, resource-bounded) symmetry of
information (see for example, \cite{LongpreWatanabe95}), but directly
construct a counterexample pair.
\begin{lemma}
\label{lem:t_incompressibility_two}
There are sequences $A$ and $B$ where $A \uplus B$ is
$t$-compressible, but $B$ is $t$-incompressible and $A$ is
$t^B$-incompressible.
\end{lemma}
\begin{proof}
We build such a pair in stages.  In the stage $s=0$, we set
$A_s=B_s=\lambda$. Then in $s \ge 1$, assume that we have inductively
defined prefixes $A_{s-1}$ of $A$ and $B_{s-1}$ of $B$, where
$|A_{s-1}|= t(s-1)$ and $|B_{s-1}|=2^{t(s-1)^2}$. We select strings
$\alpha_s$ and $\beta_s$ satisfying specific
incompressibility properties and then define
\begin{align*}
A_s  = A_{s-1} \alpha_s \quad\text{and}\quad B_s = B_{s-1}
\beta_s \alpha_s.
\end{align*}

We choose $\alpha_s$ and $\beta_s$ which satisfy following
incompressibility requirements.
\begin{enumerate}
\item Length requirements: $|\alpha_s| = t(s)-t(s-1)$ and $|\beta_s|=2^{t(s)^2}-2^{t(s-1)^2}-t(s)+t(s-1)$. These lengths ensure that $|A_s|=t(s)$ and
  $|B_s|=2^{t(s)^2}$.  
\item Incompressibility requirements for $B$: there is a constant $c$
  such that 
  $$K(B_{s-1}\delta) \ge |B_{s-1}\delta|-c$$ 
  for every $\delta  \preceq \beta_s\alpha_s$.
\item Incompressibility requirements for $A$ relative to $B$:
  $$K(A_{s-1} \tau \mid B_{s-1}\beta_s) \ge |A_{s-1}\tau|-c'$$ for
  some constant $c'$ and every $\tau \preceq \alpha_{s}$.
\end{enumerate}

It suffices to show we can find such strings $\alpha_s$ and $\beta_s$. We can select the strings in the following order. First,
select a string $\beta_s$ such that $K(\delta \mid
B_{s-1})>|\delta|-c$ for every prefix $\delta$ of $\beta_s$, and some
constant $c$.  Such a string exists, since the set of Martin-L\"of
random strings has measure $1$. Then select the string $\alpha_s$ to satisfy $K(\tau \mid A_{s-1}, B_{s-1}\beta_s) \ge |\tau| - c'$ for every prefix $\tau$ of $\alpha_s$, and a
constant $c'$. Each of these selections is possible because the set of
incompressible strings conditioned on any other strings is non-empty
(for example, see the Ample Excess Lemma \cite{MillerYu08}).

By the above construction it is clear that $B$ is $t$-incompressible
and $A \uplus B$ is $t$-compressible for any function $t(n)> n$ due to
the shared component $\alpha_s$ for all $s$ between $A$ and
$B$ as long as $t(.)$ is time-constructible which is indeed the case by our assumption in the beginning of the current section. 

We now show that $A$ is $t^B$-incompressible. Inductively, assume that
$A_{s-1}$ is $t$-incompressible given access to $B\res (2^{t(s-1)^2})$. By
construction, for every $\tau$ such that $A_{s-1} \prec \tau \preceq
A_{s}$ is $t$-incompressible given access to $B\res(2^{O(t(|\tau|))})$, since
we ensure that every prefix of $A_s$ is incompressible given $B_{s-1}
\beta_s$, i.e. $B \res (2^{t(s)^2}-t(s)+t(s-1))$, which is longer than $B \res
(2^{\omega(t(s))})$, the prefix of $B$ that $A_s$ can query within the time-bound
$t$.   \end{proof}

\section{Resource-bounded relative randomness and martingales}
\label{sbsc:rbm}
In this section, we show that the symmetry of relative randomness does
not hold for resource-bounded martingales. Let $t: \N \to \N$ be a
superlinear function. For any input $\sigma \in \Sigma^*$, we
henceforth restrict ourselves to martingales computed in time
$O(t(|\sigma|))$ and we define $t$-randomness accordingly.

\begin{definition}
A \emph{$t$-bounded martingale} is a martingale $d: \Sigma^* \to
\Rplus$ such that for all $w \in \Sigma^*$, $d(w)$ can be computed in
at most $O(t(|w|))$ steps.
\end{definition}

Unlike computably enumerable martingales, these martingales have to
terminate with the ultimate value of the bet in a finite number of
steps. The notion of success of a $t$-bounded martingale is the same
as that in the case of computably enumerable martingales.

\begin{definition}
We say that $X \in \Sigma^{\infty}$ is \emph{$t$-random} if there is
no $t$-bounded martingale which succeeds on $X$, and \emph{$t$-random
  with respect to $Y \in \Sigma^{\infty}$} if no $t$-bounded oracle
martingale $d^Z: \Sigma^* \to \Rplus$ exists such that $X \in
S^\infty[d^Y]$.
\end{definition}

\begin{lemma}
\label{lem:t_martingale_one}
There is a $t$-random sequence $B$ and a sequence $A$ which is
$t^B$-random, where $A \uplus B$ is $t$-nonrandom.
\end{lemma}
The idea of the construction is that at some positions, substrings in
$A$ are copied exactly from regions of $B$. These regions of $B$
sufficiently far so that it is not possible to consult the relevant
region in time $O(t)$. Of course, $A \uplus B$ is non-random since a
significant suffix of $B$ can be computed directly from the relevant
region of $A$.

Elsewhere, if $B$ is random, and $A$ random relative to $B$, then we
can make $B$ $t$-random, and $A$ to be $t^B$-random.

In short, the construction ensures that $B$ has sufficient time to
look into the prefix of $A$, but $A$ does not have time to look into
the extension of $B$.
\begin{proof}
We construct two sequences $A$ and $B$ in stages, where at stage
$s=0$, we have $A_s=B_s=\lambda$. At stage $s \ge 1$, let us assume
that we have inductively defined prefixes $A_{s-1}$ of $A$ and
$B_{s-1}$ of $B$ and additionally $|A_{s-1}|=t(s-1)$ and
$|B_{s-1}|=2^{t(s-1)^2}$. We select strings $\alpha_s$ and $\beta_s$ satisfying specific randomness properties and then define
\begin{align*}
A_s  &= A_{s-1} \alpha_s \quad\text{and}\quad B_s  = B_{s-1}
\beta_s \alpha_s  
\end{align*}
We choose strings $\alpha_s$ and $\beta_s$ which satisfy
all the following randomness requirements.
\begin{enumerate}
\item Length requirements: $|\alpha_s| = t(s)-t(s-1)$ and $|\beta_s|=2^{t(s)^2}-2^{t(s-1)^2}-t(s)+t(s-1)$. These lengths ensure that $|A_s|=t(s)$ and $|B_s|=2^{t(s)^2}$.  
\item Randomness requirements for $B$: for some universal martingale
  $d^{B_{s-1}}$, for every $\delta \preceq \beta_s$, $d(B_{s-1}\delta)
  \le d(B_{s-1})$.
\item Randomness requirements for $A$ relative to $B$: for some
  universal oracle martingale $d^{B_{s-1}\beta_{s}}$, for every $\tau
  \preceq \alpha_s$, $d^{B_{s-1}\beta_{s}}(A_{s-1}\tau) \le
  d^{B_{s-1}}(A_{s-1})$.
\end{enumerate}

It suffices to show we can find such strings $\alpha_s$ and $\beta_s$. We can select the strings in the following order. First,
select a string $\beta_s$ which satisfies the fact that for a
universal martingale $d$, and for every $\sigma \preceq \beta_s$,
$d(B_{s-1}\sigma) \le d(B_{s-1})$.  Such a string $\beta_s$ exists
because the martingale property together with the Markov inequality
allows us to show that for any string $\kappa$ and any $n$, the set
$\{\rho \in \Sigma^n \mid \forall \sigma \preceq\rho,\ d(\kappa\sigma)
\le d(\kappa)\}$ has positive probability. By a similar argument we
can then select the string $\alpha_s$ such that for a
universal martingale $d^{B_{s-1}\beta_s}$, and for every $\tau \preceq
\alpha_s$, $d^{B_{s-1}\beta_s}(A_{s-1}\tau) \le
d^{B_{s-1}}(A_{s-1})$.

By construction it is clear that $B$ is Martin-L\"of random and $A
\uplus B$ is not $t$-random for any superlinear $t$ due to the shared
component $\alpha_s$ for all $s$ between $A$ and $B$ as long as the function $t(.)$ is time-constructible which is indeed the case by our assumption. However we can
show that $A$ is $t^B$-random. By the construction it can be noted
that any martingale $d^B$, can gain capital on the stretch $\alpha_s$
only if it can query the corresponding portion of the sequence $B$. To
calculate the value of $d^B(A \res n)$ it needs to query the index
bigger than $2^{\omega(t(n))}$ of the sequence $B$, which is
impossible in the given time-bound.  \end{proof}

Now, we consider the converse.

\begin{lemma}
Let $B$ be an arbitrary $t$-random sequence. Then there is a
$t$-random sequence $A$ which is $t^B$-nonrandom, such that $B$ is
$t^A$-random.
\end{lemma}
\begin{proof}
  (Sketch) The construction is similar to that of Lemma
  \ref{lem:t_martingale_one}. Let $B$ be a Martin-L\"of random
  sequence, and $A$ be a Martin-L\"of random sequence, except for a
  short string at $A_{2^{t(n)^2}}$ identical to a string at $B_n$, $n
  \in \N$. We see that $A$ is $t^B$-nonrandom. However, $B$ does not
  have sufficient time to consult the relevant position in $A$, and is
  $t^A$-random.   \end{proof} However, in the above example, $A
\uplus B$ is $t$-nonrandom, since the substring at $(A \uplus B)_{2
  \times 2^{t(n)^2}}$ is computable from the prefix of $(A \uplus B)
\res 2n$. Thus the identification of relative randomness of $A$ and
$B$ with the randomness of $A \uplus B$ breaks down in time-bounded
settings.
\begin{corollary}
There are sequences $A$ and $B$ such that $A$ is $t^B$ nonrandom, $A
\uplus B$ is $t$-nonrandom, and $B$ is $t^A$ random.
\end{corollary}
Now let us first make an observation.
\begin{lemma}
If $B$ is $t$-nonrandom then for any sequence $A$, $A\uplus B$ is
$t$-non-random.
\end{lemma}
\begin{proof}
If $d_B$ be a $t$-martingale witnesses the fact that $B$ is
$t$-nonrandom, then the martingale $d_{AB}$ defined in (1) is
a $t$-martingale that succeeds on $A\uplus B$.
  
\end{proof}
We wish to investigate the question of $t$-randomness of $A \uplus B$
given that $A$ is $t^B$-nonrandom. We have weak converses which we now
describe. The above corollary suggests that we stipulate ``honest''
reductions - that a bit at position $n$ in $A$ cannot depend on bits
at positions $o(t^{-1}(n))$ in $B$. With this stipulation, we have the
following weak converse to Lemma \ref{lem:t_martingale_one}. First, we
consider a restricted class of reductions from $A$ to $B$.
\begin{definition}
We say that an infinite sequence $A$ is \emph{infinitely often
  reducible to} $B$ in time $t$ via $f$, written $A \le^t_{i.o} B$, if
$\{ n \in \N \mid f(B[n-t(n) \dots n+t(n)-1]) = A_n\}$ is computable
in time $O(t(n))$, i.e., $t$-computable.
\end{definition}
Note that we have incorporated an honesty requirement into the
definition. 
\begin{definition}
We say that a function $f : \Sigma^* \to \Sigma$ is \emph{strongly
  influenced by the last index} if 
for every $\sigma \in \Sigma^n$, $f(\sigma) \ne f((\sigma \res
n-1)  \overline{\sigma_n})$.
\end{definition}
The function that projects the last bit of its input, and the function
computing the parity of all input bits are two examples of such
functions. 
\begin{lemma}
Let $B$ be $t$-random and $A \le^t_{i.o} B$ via a function that is
strongly influenced by the last index. Then $A \uplus B$ is also
$t$-nonrandom. 
\end{lemma}
\begin{proof}
Consider the $t$-computable set of positions $S = \{ n \mid f(B[n-t(n)
  \dots n+t(n)-1]) = A_n \}$ where $A$ queries $B$. We define a
martingale $d$ with initial capital 1 and which bets evenly on all
positions except those in the set $T$ defined by
$$T = \{ 2 (i+t(i)) + 1 \mid i \in S\}.$$ 
For positions $2(i+t(i))+1 \in T$, sets $d(A_0 B_0 \dots A_{i+t(i)}b)$
to $2 d(A_0 \dots A_{i+t(i)})$ if $f((B \res i+t(i)-1) b) = A_i$, and
to $0$ otherwise. Then $A \uplus B \in S^\infty[d]$.  
\end{proof}
A second weak converse can be obtained by assuming that the
$t$-martingale succeeds on the interleaved sequence in a specific
manner.
\begin{definition}
We say that a pair of sequences $(A,B)$ is $t$-resilient if   
\begin{enumerate}
\item For every oracle martingale $h$ runs in time $O(t(n))$,
  $\limsup_{n\to \infty} h^{B\res n-1}(A\res n) < \infty$.
\item For every oracle martingale $g$ runs in time $O(t(n))$,
  $\limsup_{n \to \infty} g^{A\res n}(B \res n) < \infty$.
\end{enumerate}
\end{definition}



We say that a martingale $d$ wins at position $i$ on a sequence $X$
if $d(X \res i) > d(X \res i-1)$. 
 
\begin{lemma}
$A\uplus B$ is $t$-random iff $(A,B)$ is a $t$-resilient pair.
\end{lemma}
\begin{proof}
Suppose that there exists a martingale $d$ which runs in time
$O(t(n))$ and witnesses the fact that $A\uplus B$ is restricted
$t$-nonrandom. Now construct the oracle martingales $h$ and $g$ as
follows:
\begin{align*}
h^Y(\sigma)=g(\sigma)&=d(\lambda) \ \text{if } \sigma= \lambda \text{
  or }\sigma\in\Sigma\\ 
h^Y(X\res n)&=\frac{d(X\uplus Y\res 2n-1)}{d(X\uplus Y \res
  2n-2)}\cdot h^Y(X\res n-1)\\ 
g^X(Y\res n)&=\frac{d(X\uplus Y\res 2n)}{d(X\uplus Y \res 2n-1)}\cdot
g^X(Y\res n-1)\\ 
\end{align*} 
Clearly $h$ is dependent on $B\res n-1$ and $g$ is dependent on $A\res
n$. Since $\limsup_{n\to\infty}d(A\uplus B\res n)\uparrow \infty$, we
claim that one of $h$ and $g$ succeeds over $A$ and $B$ given $B\res
n-1$ and $A\res n$ respectively.  We have
\begin{multline*}
\limsup_{n\to\infty} h^{B\res n-1}(A\res n) \cdot g^{A\res n}(B\res n)
=\limsup_{n\to\infty} \frac{d(A\uplus B \res 2n)}{c}\\ 
\leq \limsup_{n\to\infty} h^{B\res n-1}(A\res
n)\cdot \limsup_{n\to\infty} g^{A\res n}(B\res n)
\end{multline*}
for some fixed constant $c$ (independent of $n$).  Note that LHS is
$\infty$ because $d(A\uplus B\res n)$ is a sequence which satisfies
the property $$2d(A\uplus B\res n-1)\geq d(A\uplus B\res n)\geq0$$ and
$\limsup_{n\to\infty}d(A\uplus B\res n)= \infty$. So one of the term
involving $h$ or $g$ has to go to $\infty$. Now we show that $h$ and
$g$ are oracle martingales which run in time $O(t(n))$. By
construction $h$ and $g$ are oracle functions computable in time
$O(t(n))$. Now
\begin{multline*}
\sum_{b \in \Sigma} h((A \res n) b) 
= \frac{h(A\res n)}{d(A\uplus B\res 2n)}\sum_{b\in \Sigma}d((A\uplus
B\res 2n)b) 
= 2\cdot h(A\res n)
\end{multline*}
and thus $h$ is a oracle martingale. By a similar argument $g$ will
also become a oracle martingale which runs in time $O(t(n))$. Since
either $\limsup_{n\to\infty}h^{B\res n-1}(A\res n)= \infty$ or
$\limsup_{n\to\infty}g^{A\res n}(B\res n)= \infty$, it follows that
$(A,B)$ is a not a $t$-resilient pair.

Conversely, if $(A,B)$ is not a $t$-resilient pair then either there
is a oracle martingale $h$ runs in time $O(t(n))$ such that
$\limsup_{n\to\infty}h^{B\res n-1}(A\res n)= \infty$, or a oracle
martingale $g$ runs in time $O(t(n))$ such that $\limsup_{n \to
  \infty} g^{A\res n}(B \res n) = \infty$. If the first condition
holds, then
\begin{align*}
d(A\uplus B\res 2n-1)&=h^{B\res n-1}(A\res n),\\
d(A\uplus B\res 2n)&= d(A\uplus B\res 2n-1)
\end{align*}
is a $t$-martingale witnessing that $A \uplus B$ is $t$-nonrandom. If
the second condition holds then we can define
a similar martingale $d$ based on $g$, witnessing $t$-nonrandomness of
$A \uplus B$.
 \end{proof}

\section{A Modified Definition of Resource-bounded Martingales}
In this section, we propose an alternate definition of a time-bounded
martingale whose behavior with respect to van Lambalgen's theorem
is identical to the definition using time-bounded
prefix-complexity. In the light of van Lambalgen's theorem, we may
view this as a reasonable variant definition.

\begin{definition}
We say that a martingale $d: \Sigma^* \to \Rplus$ is a
\emph{$t$-bounded lookahead} martingale if there is for each string
$\sigma$, there is a set $L_{d,\sigma} \subseteq \N$ such that the
following conditions are satisfied.
\begin{enumerate}
\item $d(\lambda) = 1$ and $L_{d,0} = \emptyset$.
\item For any string $\sigma$, $d(\sigma  0) + d(\sigma 
  1) = 2d(\sigma)$.
\item For any string $\sigma \in \Sigma^{n-1}$, if $\begingroup
  n \endgroup \not\in L_{d,\sigma}$ then to compute
  $d(\sigma b)$, $b \in \Sigma$, the martingale can query a set of
  positions $S \subseteq \{0, \dots, n-2, n, \dots,
  O(t(n))\}$. Subsequently, $L_{d,\sigma b}$ is set to $L_{d,\sigma}
  \cup S$. If $n-1 \in L_{d,\sigma}$, then we forbid betting, and set
  $d(\sigma b)$ to $d(\sigma)$, and $L_{d,\sigma b}$ to $L_{d,\sigma}$.
\end{enumerate}
\end{definition}

\begin{definition}
  We say that an infinite sequence is \emph{$t$-lookahead-non-random}
  if there is a $t$-bounded lookahead martingale which succeeds on it.
\end{definition}

To compute $d(X \res n)$, the martingale is allowed to wait until an
appropriate extension length is available, and base its decision on a
few bits ahead. However, we have to be careful not to reveal $X_{n-1}$
itself, and to ensure that positions once revealed can never later be
bet on. These restrictions ensure that the betting game is not
trivial, and that there are unpredictable or random sequences.

\begin{lemma}
There is a $t$-lookahead random sequence $B$ and a $t^B$-lookahead
random sequence $A$ such that $A \uplus B$ is $t$-lookahead nonrandom.
\end{lemma}
The proof is essentially the same as that of Lemma
\ref{lem:t_martingale_one}.

With the modified definition, we can now prove result similar to Lemma
\ref{lem:t_incompressibility_one}. 
\begin{lemma}
If $B$ is $t$-lookahead nonrandom or $A$ is $t^B$-lookahead
nonrandom. Then $A \uplus B$ is $t$-lookahead nonrandom.
\end{lemma}
\begin{proof}
Suppose $h$ is a $t$-lookahead-martingale that succeeds on $B$. Then
define the $t$-lookahead martingale $d$ by setting $d(\lambda)=1$ and
$L_{d,\lambda} = \emptyset$, and
\begin{align*}
d((X \uplus Y) \res 2n+1) &= d((X \uplus Y) \res 2n), \quad
  L_{d,(X \uplus Y) \res
    2n+1} = L_{d,(X \uplus Y) \res 2n}\\
d((X \uplus Y) \res 2n+2) &= h(Y \res n), \quad L_{d,(X \uplus Y) \res
    2n+2} = \{2i+1 | i \in L_{h, Y \res n}\}.
\end{align*}
Then clearly $A \uplus B \in S^\infty[d]$ as $B \in S^\infty[h]$.

Now, assume that $A \in S^\infty[g^B]$ for a $t$-lookahead martingale
$g$. Then we define the $t$-lookahead martingale $d$ by $d(\lambda)=1$
with $L_{d,\lambda} = \emptyset$ and 
\begin{align*}
d((X \uplus Y) \res 2n+2) &= d((X \uplus Y) \res 2n+1), \quad L_{d,(X
  \uplus Y) \res 2n+2} = L_{d,(X \uplus Y) \res 2n+1}\\ 
d((X \uplus Y) \res 2n+1) &= g^{Y \res O(t(n))}(X \res n),\\
           L_{d,(X \uplus Y) \res 2n+1} 
           &= L_{d,(X \uplus Y) \res 2n} 
               \cup \{2i|i\in L_{g, X \res n}\}
               \cup \{2i+1|i \in Q_{g, X, Y, n}\},
\end{align*}
where $Q(g,X,Y,n)$ are the bits in the oracle queried by $g^Y(X \res
n)$. We know that $A \in S^\infty[g^B]$. Hence $A \uplus B \in
S^\infty[d]$. 
 \end{proof}

\section*{Acknowledgments}
The authors thank Jack Lutz, Manjul Gupta and Michal Kouck{\'y} for
helpful discussions.

\bibliography{main,dim,fair001}
\end{document}